\documentclass[a4paper]{article}
\pdfoutput=1

\usepackage{amsmath, amsthm, amssymb}
\usepackage{graphicx,color} 

\usepackage{authblk}

\usepackage{enumitem}

\theoremstyle{plain} 
\newtheorem{theorem}{Theorem}
\newtheorem{lemma}[theorem]{Lemma}
\newtheorem{fact}{Fact}
\newtheorem{cor}[theorem]{Corollary}
\newtheorem{prop}[theorem]{Proposition}

\begin{document}

\title{Uniform 2D-Monotone Minimum Spanning Graphs\footnote{This research was financially supported by the Special Account for Research Grants of the National Technical University of Athens.}}

\author{Konstantinos Mastakas} 
\affil{School of Applied Mathematical and Physical Sciences\\ National Technical University of Athens, Athens, Greece \\ kmast@math.ntua.gr}

\date{}

\maketitle

\begin{abstract}
A geometric graph $G$ is \emph{$xy-$monotone} if each pair of vertices of $G$ is connected by a $xy-$monotone path.
We study the problem of producing the $xy-$monotone spanning geometric graph of a point set $P$ that (i) has the minimum cost, where the cost of a geometric graph is the sum of the Euclidean lengths of its edges, and (ii) has the least number of edges, in the cases that the Cartesian System $xy$ is specified or freely selected.
Building upon previous results, we easily obtain that the two solutions coincide when the Cartesian System is specified and are both equal to the rectangle of influence graph of $P$.
The \emph{rectangle of influence graph} of $P$ is the geometric graph with vertex set $P$ such that two points $p,q \in P$ are adjacent if and only if the rectangle with corners $p$ and $q$ does not include any other point of $P$.
When the Cartesian System can be freely chosen, we note that the two solutions do not necessarily coincide, however we show that they can both be obtained in $O(|P|^3)$ time.  
We also give a simple $2-$approximation algorithm for the problem of computing the spanning geometric graph of a $k-$rooted point set $P$, in which each root is connected to all the other points (including the other roots) of $P$ by $y-$monotone paths, that has the minimum cost.
\end{abstract}

\section{Introduction} 

A sequence of points in the Euclidean plane $q_0$, $q_1$, \ldots, $q_t$ is called \emph{$y-$monotone} if the sequence of their $y$ coordinates, i.e.{} $y(q_0)$, $y(q_1)$, \ldots, $y(q_t)$,
is either decreasing or increasing, with $y(p)$ denoting the $y$ coordinate of the point $p$.  
A geometric path $Q$ $=$ $(q_0$, $q_1$, \ldots, $q_t)$ is called \emph{$y-$monotone} if the sequence of its vertices, i.e.{} the sequence $q_0$, $q_1$, \ldots, $q_t$, is $y-$monotone. 
If $Q$ is $y'-$monotone for some axis $y'$ then $Q$ is called \emph{monotone}.
Let $G = (P,E)$ be a geometric graph.
If each $p,q \in P$ are connected by a $y-$monotone path then $G$ is called \emph{$y-$monotone}.
If $G$ is $y'-$monotone for some axis $y'$ then $G$ is called \emph{uniform monotone} (following the terminology of~\cite{MasS17}).
Uniform monotone graphs were called \emph{$1-$monotone graphs} by Angelini~\cite{Ang17}.
If each $p,q \in P$ are connected by a monotone path, where the direction of monotonicity might differ for different pairs of vertices, then $G$ is called \emph{monotone}.
Monotone graphs were introduced by Angelini et al.~\cite{AngCBFP12}.
Drawing an (abstract) graph as a monotone (geometric) graph has been a topic of research~\cite{Ang17,AngCBFP12,AngDKMRSW15,HeH17,OikS17}.

The \emph{Monotone Minimum Spanning Graph problem}, i.e.{} the problem of constructing the monotone spanning geometric graph of a given point set that has the minimum cost, where the cost of a geometric graph is the sum of the Euclidean lengths of its edges, was recently introduced (but not solved) in~\cite{MasS17} and it remains an open problem whether it is NP-hard.
Since the more general (without the requirement of monotonicity) \emph{Euclidean Minimum Spanning Tree problem} can be solved in $\Theta(|P| \log |P|)$ time~\cite{ShaH75}, this constitutes a great differentiation that is induced by the addition of the property of monotonicity.

A point set $P$ is \emph{$k-$rooted} if there exist $k$ points $r_1$, $r_2$, \ldots, $r_k \in P$ distinguished from the other points of $P$ which are called the \emph{roots} of $P$.
A geometric graph $G = (P,E)$ is called \emph{$k-$rooted} if $P$ is $k-$rooted and its roots are the roots of $P$.
A $k-$rooted geometric graph $G$ is \emph{$k-$rooted $y-$monotone} if each root $r \in P$ and each point $p \in P\setminus\{r\}$ are connected by $y-$monotone paths.
Similarly, $G$ is \emph{$k-$rooted uniform monotone} (following the terminology of~\cite{MasS17}) if it is $k-$rooted $y'-$monotone for some axis $y'$.
For simplicity, we may also denote point sets or geometric graphs that are $1-$rooted simply as rooted.
A polygon that is $2-$rooted $y-$monotone, in which its roots are its lowest and highest vertices, can be triangulated in linear time~\cite{GarJPT78}.
Lee and Preparata~\cite{LeeP77} preprocessed a subdivision $S$ of the plane such that the region in which a query point belongs can be found quickly, by (i) extending the geometric graph bounding $S$ to a $2-$rooted $y-$monotone planar geometric graph in which the roots are the highest and lowest vertices of $S$, and (ii) constructing a set of appropriate $y-$monotone paths from the lowest to the highest vertex of $S$.
Additionally, Lee and Preparata~\cite{LeeP77} noted that a $2-$rooted planar geometric graph, where all vertices have different $y$ coordinates, in which the roots are the highest and lowest vertices of the graph is $2-$rooted $y-$monotone if and only if each non-root vertex has both a neighbor above it and a neighbor below it.
Furthermore, a rooted geometric graph $G = (P,E)$, where all vertices have different $y$ coordinates, with a (single) root $r$ that is not the highest or lowest point of $P$ is rooted $y-$monotone if and only if each non-root vertex $p$ has a neighbor $q$ such that $y(q)$ is between $y(r)$ (inclusive) and $y(p)$~\cite{MasS17}.
Additionally, rooted uniform monotone graphs can be efficiently recognized~\cite{MasS17}.
The \emph{$k-$rooted $y-$monotone (uniform monotone) minimum spanning graph} (following the terminology of~\cite{MasS17}) of a $k-$rooted point set $P$ is the $k-$rooted $y-$monotone (uniform monotone) spanning graph of $P$ that has the minimum cost.
The rooted $y-$monotone (uniform monotone) minimum spanning graph\footnote{In~\cite{MasS17} it is shown that it is actually a tree.} of a rooted point set $P$ can be produced in $O(|P|\cdot \log^2 |P|)$ (resp.{}, $O(|P|^2\cdot \log |P|)$) time~\cite{MasS17}.
The problem of drawing a rooted tree as a rooted $y-$monotone minimum spanning graph is studied in~\cite{Mas18}.
 The ($|P|-$rooted) $y-$monotone minimum spanning graph of a point set $P$ is the geometric path that traverses all the points of $P$ by moving north, from the lowest point to the highest point of $P$~\cite{MasS17}.
Regarding the problem of producing the $k-$rooted $y-$monotone minimum spanning graph of a $k-$rooted point set $P$, with $1 < k < |P|$, it is an open problem, posed in~\cite{MasS17}, whether it is NP-hard. 

The restricted fathers tree problem was introduced in~\cite{Gut-BecH10} and is related to the rooted $y$-monotone minimum spanning graph problem constrained to rooted point sets $P$ in which the $y$ coordinate of the root is zero and the $y$ coordinates of the other points of $P$ are all negative (or all positive). 
The input of the \emph{restricted fathers tree problem} is a complete graph with root where each edge has a cost and each vertex has a value and the goal is to output the spanning tree in which the path from the root to each vertex decreases in value that has the minimum cost.
The restricted fathers tree problem is greedily solvable~\cite[Corollary 2.6]{Gut-BecH10}.

A geometric path $Q$ $=$ $(q_0$, $q_1$, \ldots, $q_t)$ is \emph{$xy-$monotone} if the sequence of its vertices is both $x-$monotone, i.e.{} the sequence $x(q_0)$, $x(q_1)$, \ldots, $x(q_t)$,
is monotone, and $y-$monotone.
$Q$ is \emph{2D-monotone} (following the terminology of~\cite{MasS17}) if it is $x'y'-$monotone for some orthogonal axes $x',y'$.
A geometric graph $G = (P,E)$ is \emph{2D-monotone} (following the terminology of~\cite{MasS17}) if each pair of points of $P$ is connected by a 2D-monotone path.
2D-monotone paths/graphs were called \emph{angle-monotone paths/graphs} by Bonichon et al.~\cite{BonBCKLV16}.
Bonichon et al.~\cite{BonBCKLV16} showed that deciding if a geometric graph $G = (P,E)$ is 2D-monotone can be done in $O(|P|\cdot |E|^2)$ time.
Triangulations with no obtuse internal angles are 2D-monotone graphs~\cite{DehFG15,Lub_oR17}.
There exist point sets for which any 2D-monotone spanning graph is not planar~\cite{BonBCKLV16}.
The problem of constructing 2D-monotone graphs with asymptotically less than quadratic edges was studied by Lubiw and Mondal~\cite{LubM18}.
It is an open problem, posed in~\cite{MasS17}, whether the 2D-monotone spanning graph of a point set $P$ that has the minimum cost can be efficiently computed. 

The \emph{(rooted) $xy-$monotone} and \emph{(rooted) uniform 2D-monotone} (using the terminology of~\cite{MasS17}) graphs are defined similar to the (rooted) $y-$monotone and (rooted) uniform monotone graphs.
Deciding if a rooted geometric graph $G = (P,E)$ is rooted $xy-$monotone (uniform 2D-monotone) can be done in $O(|E|)$ (resp.{}, $O(|E|\cdot \log |P|)$) time~\cite{MasS17}.
Additionally, the rooted $xy-$monotone (uniform 2D-monotone) spanning graph of a rooted point set $P$ that has the minimum cost\footnote{In~\cite{MasS17} it is shown that it is actually a tree, denoted as the \emph{rooted $xy-$monotone (uniform 2D-monotone) minimum spanning tree} in~\cite{MasS17} and abbreviated as the \emph{rooted $xy-$MMST} (resp.{}, \emph{rooted $2D-$UMMST}) in~\cite{MasS17}.} can be computed in $O(|P|\cdot \log^3|P|)$ (resp.{}, $O(|P|^2\log|P|)$) time~\cite{MasS17}. 
We focus on the production of the \emph{$xy$-monotone minimum spanning graph ($xy-$MMSG)} of a point set $P$, i.e.{} the $xy$-monotone spanning graph of $P$  that has the minimum cost, and the production of the \emph{uniform $2D-$monotone minimum spanning graph ($2D-$UMMSG)} of a point set $P$, i.e.{} the uniform $2D-$monotone spanning graph of $P$ that has the minimum cost.
 We also study the corresponding problems regarding the production of the spanning graphs with the least number of edges, i.e.{} the production of the $xy$-monotone spanning graph with the least number of edges and the production of the uniform $2D-$monotone spanning graph with the least number of edges.

A curve $C$ is \emph{increasing-chord}~\cite{LarM72,Rot94} if for each $p_1, p_2, p_3, p_4$ traversed in this order along it, the length of the line segment $\overline{p_1p_4}$ is greater than or equal to the length of $\overline{p_2p_3}$.
Alamdari et al.~\cite{AlaCGLP13} introduced \emph{increasing-chord graphs} which are the geometric graphs for which each two vertices are connected by an increasing-chord path.
Increasing-chord graphs are widely studied~\cite{AlaCGLP13,BahDMM17,DehFG15,MasS15,NolPR16}.
The problem of producing increasing-chord spanning graphs (where Steiner points may be added) of a point set $P$ was studied in~\cite{AlaCGLP13,DehFG15,MasS15}.
The approach employed in~\cite{AlaCGLP13,DehFG15,MasS15}, was to connect the points of $P$ by 2D-monotone paths since as noted by Alamdari et al.~\cite{AlaCGLP13} 2D-monotone paths are also increasing-chord paths.

Let $P$ be a point set and let $p,q \in P$ then $p$ and $q$ are \emph{rectangularly visible} if the rectangle with corners $p$ and $q$ does not include any other point of $P$.
Furthermore, the \emph{rectangle of influence graph} of $P$ is the geometric graph spanning $P$ such that $\overline{pq}$ is an edge of the graph if and only if $p$ and $q$ are rectangularly visible.
Alon et al.~\cite{AloFK85} denoted rectangularly visible points as~\emph{separated points} and the rectangle of influence graph as the \emph{separation graph}.
Computing the rectangle of influence graph $G = (P,E)$ of $P$ can be done in $O(|P|\cdot \log |P| + |E|)$ time~\cite{OveW88}.
There exist point sets $P$ for which the number of edges of their rectangle of influence graph is $\Omega (|P|^2)$~\cite{AloFK85}.
The rectangle of influence graph does not remain the same if the Cartesian System is rotated~\cite[Proposition 3]{IchS85}.
Drawing an abstract graph as a rectangle of influence graph has been studied~\cite{LioLMW98}.

\noindent \textbf{Our Contribution.}
Building upon previous results, we easily obtain that given a point set $P$ the $xy-$MMSG of $P$ is equal to the $xy-$monotone spanning graph of $P$ that has the least number of edges and are both equal to the rectangle of influence graph of $P$.
We note that given a point set $P$ the $2D-$UMMSG of $P$ does not necessarily coincide with the uniform 2D-monotone spanning graph of $P$ that has the least number of edges.
We also show that both the $2D-$UMMSG of $P$ and the uniform 2D-monotone spanning graph of $P$ that has the least number of edges can be produced in $O(|P|^3)$ time.    
Additionally, we give a simple $2-$approximation algorithm for the problem of producing the $k-$rooted $y-$monotone minimum spanning graph of a $k-$rooted point set. 

\section{Preliminaries}

\subsection{$xy-$Monotone Minimum Spanning Graphs}

 Angelini~\cite{Ang17} noted the following Fact regarding $y-$monotone graphs.

\begin{fact}[Angelini~\cite{Ang17}]
\label{fact:yMonGraphNecessaryEdges}
Let $G = (P,E)$ be a $y-$monotone graph where no two points of $P$ have the same $y$ coordinate and let $p,q \in P$ such that for each $r \in P\setminus\{p,q\}$ the sequence $p,r,q$ is not $y-$monotone.
Then, $p$ and $q$ are adjacent in $G$.
\end{fact}

Fact~\ref{fact:yMonGraphNecessaryEdges} is easily extended in the context of $xy-$monotone graphs.
More specifically, let $G = (P,E)$ be a $xy$-monotone graph and $p, q \in P$ such that for each $r \in P\setminus\{p,q\}$ the sequence of points $p,r,q$ is not $xy-$monotone, then $p$ and $q$ are adjacent in $G$.
Alon et al.~\cite{AloFK85} noted that the points $p,q$ of a point set $P$ are rectangularly visible if and only if for each $r \in P\setminus\{p,q\}$ the sequence of points $p,r,q$ is not $xy-$monotone.
Hence, the rectangle of influence graph of $P$ is a subgraph of $G$.

Liotta et al.~\cite[Lemma 2.1]{LioLMW98} showed that the rectangle of influence graph of a point set is a $xy-$monotone graph\footnote{Technically speaking, Liotta et al.~\cite{LioLMW98} showed that the rectangle of influence graph of a point set is a graph such that each two vertices are connected by a path lying inside the rectangle defined by these vertices but upon careful reading the path that is obtained in their proof is $xy-$monotone.}.

From the previous two sentences, regarding the rectangle of influence graph, we obtain the following Corollary. 

\begin{cor} \label{cor:minSpanXy}
Let $P$ be a point set. The $xy-$MMSG of $P$ and the $xy-$monotone spanning graph of $P$ that has the least number of edges coincide and they are both equal to the rectangle of influence graph of $P$. 
\end{cor}

We recall that the rectangle of influence graph $G = (P,E)$ of $P$ can be produced in $O(|P|\cdot \log |P| + |E|)$ time~\cite{OveW88} which is optimal~\cite{OveW88} and that there exist point sets $P$ for which the rectangle of influence graph has size $\Omega (|P|^2)$~\cite{AloFK85} as well as point sets for which it has linear size~\cite{AloFK85}.

\subsection{Rooted Uniform 2D-Monotone Graphs}
Mastakas and Symvonis~\cite{MasS17} studied the problem of recognizing rooted uniform 2D-monotone graphs. 
They initially noted the following Fact.

\begin{fact}[Observation 8 in~\cite{MasS17}] \label{fact:recRootEvents}
Let $G$ be a geometric graph $G = (P,E)$ with root $r$. If one rotates a Cartesian System $x'y'$, then $G$ may become rooted $x'y'-$monotone while previously it was not, or vice versa, only when the $y'$ axis becomes (or leaves the position where it previously was) parallel or orthogonal to
\begin{enumerate}[nosep]
\item a line passing through $r$ and a point $p \in P\setminus\{r\}$.
\item an edge $\overline{pq} \in E$, where $p,q \neq r$.
\end{enumerate}
\end{fact}

Based on Fact~\ref{fact:recRootEvents}, Mastakas and Symvonis~\cite{MasS17} gave a rotational sweep algorithm denoted as the \emph{rooted uniform 2D-monotone recognition algorithm} in~\cite{MasS17}.

\begin{fact}[\cite{MasS17}] \label{fact:rU2Dm-recCost}
The rooted uniform 2D-monotone recognition algorithm \begin{enumerate}[nosep,label=\roman*)]
\item computes, in $O(|E| \cdot \log |P|)$ time, a set of sufficient Cartesian Systems, of size $O(|E|)$ , which are associated with (1) lines passing through $r$ and a point $p \in P\setminus\{r\}$ and (2) edges $\overline{pq} \in E$, where $p,q \neq r$.
\item tests, in $O(|E|)$ total time\footnote{Technically speaking in~\cite{MasS17} it is shown that the remaining steps, i.e.{} the steps after the computation of the sufficient Cartesian Systems, of the rooted uniform 2D-monotone recognition algorithm take $O(|E|\cdot \log |P|)$ total time. 
Internally in the rooted uniform 2D-monotone recognition algorithm given in~\cite{MasS17}, for each $p \in P\setminus\{r\}$ it is stored the set of adjacent points to $p$ that are in the rectangle  w.r.t.{} the Cartesian System $x'y'$ with corners $p$ and $r$, which is denoted as $A(p,x',y')$ in~\cite{MasS17}.
Furthermore, it is stored the set of points $p \in P\setminus\{r\}$ for which $|A(p,x',y')| > 0$ which is denoted as $B(x',y')$ in~\cite{MasS17}.
However, only the cardinalities of these sets are necessary~\cite[Lemma 9]{MasS17}, hence if instead of the sets $A(p,x',y'), p \in P$ and $B(x',y')$ their cardinalities are stored, the remaining steps of the rooted uniform 2D-monotone recognition algorithm take $O(|E|)$ total time.
},  
if $G$ is rooted $x'y'-$monotone for some Cartesian System $x'y'$ in the previously computed set of sufficient Cartesian Systems.
\end{enumerate}
\end{fact}

\begin{fact}[Theorem 1 in~\cite{MasS17}] \label{fact:ryMMSG_computation}
Let $P$ be a rooted point set then the rooted $y-$monotone minimum spanning graph of $P$ can be obtained in $O(|P|\cdot \log^2 |P|)$ time.
\end{fact}

\section{The 2D-UMMSG Problem} \label{sec:2D-UMMSGp}

We now deal with the construction of the $2D-$UMMSG and the uniform $2D-$monotone spanning graph with the least number of edges.
We initially show that the $2D-$UMMSG of a point set $P$ can be obtained in $O(|P|^3)$ time.
For this, we employ a rotational sweep technique.
Our approach regarding the construction of the $2D-$UMMSG is similar to the approach employed for the calculation of the rooted uniform $2D-$monotone spanning graph that has the minimum cost in~\cite{MasS17}.
We assume that no three points of $P$ are collinear and no two line segments $\overline{pq}$ and $\overline{p'q'}$, $p,p',q,q' \in P,$ are parallel or orthogonal.

Let $P$ be a point set and $p$ be a point of $P$.
Let $RV(p,x',y')$ denote the subset of points of $P$ that are rectangularly visible from $p$ w.r.t.{} the Cartesian System $x'y'$.
See for example, Figure~\ref{fig:xyRecVis}(a).

\begin{prop} \label{prop:events}
If we rotate a Cartesian System $x'y'$ counterclockwise, then the $x'y'-$MMSG of $P$ changes only when $y'$ reaches or moves away from a line perpendicular or parallel to a line passing through two points of $P$.
\end{prop}

\begin{proof}
If we rotate the Cartesian System $x'y'$ counterclockwise then the $RV(p,x',y')$ for a point $p \in P$ changes only when $y'$ reaches or moves away from a line perpendicular or parallel to a line passing through two points of $P$; e.g.{} see Figure~\ref{fig:xyRecVis}.
From the previous and since the $RV(p,x',y'), p \in P$, equals to the set of adjacent vertices of $p$ in the $x'y'-$MMSG of $P$ (Corollary~\ref{cor:minSpanXy}), we obtain the Proposition. 
\end{proof}

\begin{figure}[htbp] 
\centering
\begin{tabular}{ccc}
\includegraphics[scale=0.35]{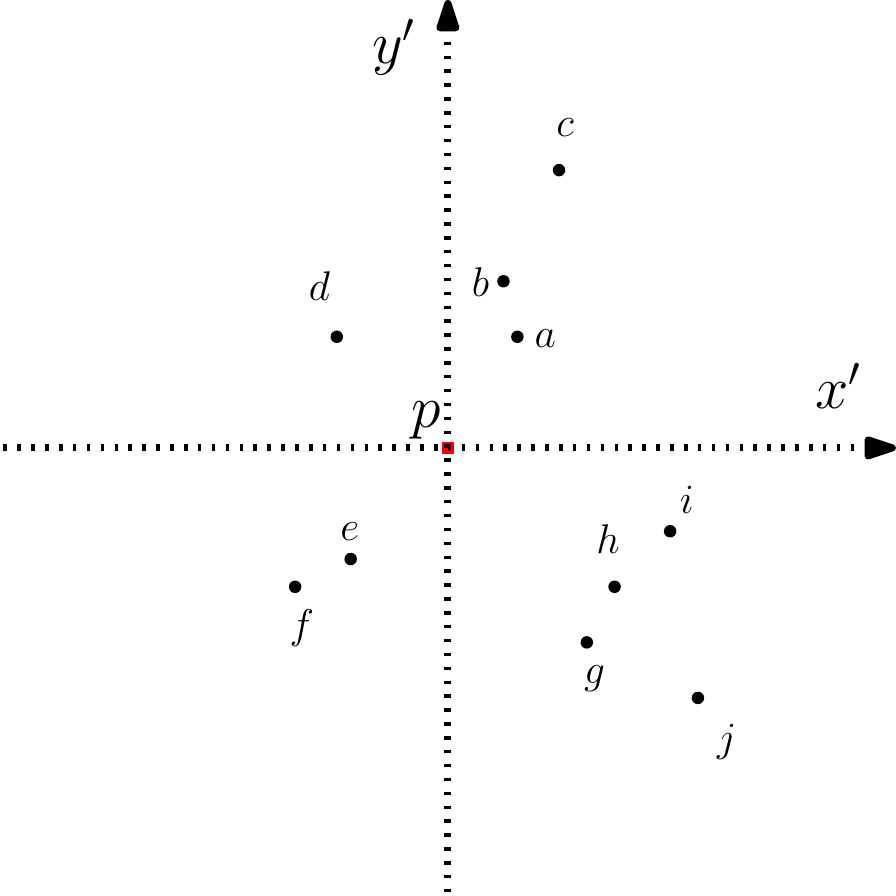} &
\includegraphics[scale=0.35]{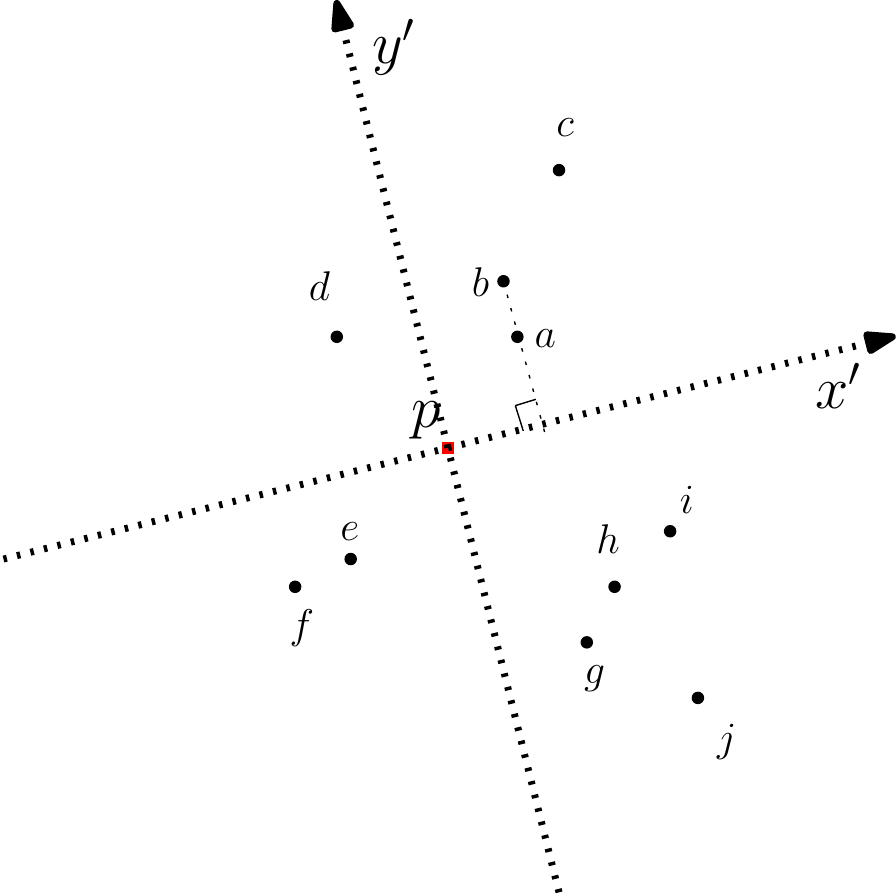} & 
\includegraphics[scale=0.35]{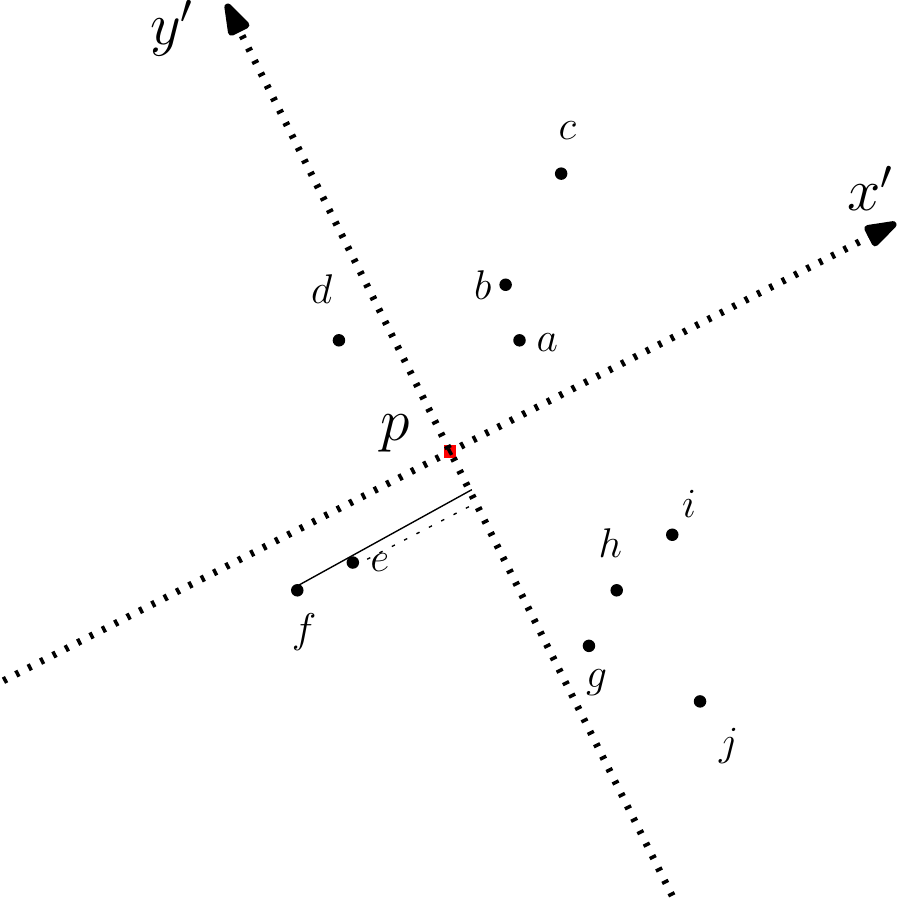}
 \\
(a) & (b) & (c)   
\end{tabular}
\caption{
In (a) $RV(p,x',y')$ $=$ $\{a$, $b$, $d$, $e$, $g$, $h$, $i\}$. In (b) the $y'$ becomes parallel to the  $\overline{ab}$ and now $b$ is not rectangularly visible from $p$.
Finally, in (c) the $y'$ has left the position where it previously was orthogonal to the $\overline{ef}$ and now $f$ becomes rectangularly visible from $p$. 
}\label{fig:xyRecVis}
\end{figure}

Let $S$ $=$ $\{s \in [0, \frac{\pi}{2}):$ a line of slope $s$ is perpendicular or parallel to a line passing through two points of $P\}$.
Let $S$ $=$ $\{s_1$, $s_2$, \ldots, $s_l\}$ with $l = \binom{|P|}{2}$ such that $0$ $\leq$ $s_1$ $<$ $s_2$ $<$ \dots $<$ $s_l$ $<$ $\frac{\pi}{2}$.
We now define the set $S_{\text{sufficient}}$ to be equal to $\{s_1$, $\frac{s_1+s_2}{2}$, $s_2$, $\frac{s_2 + s_3}{2}$, \ldots, $s_l$, $\frac{s_l +\frac{\pi}{2}}{2}\}$.
Let $x_1y_1$, $x_2y_2$, \ldots, $x_{2l}y_{2l}$ be the Cartesian Systems in which the vertical axis has slope in $S_{\text{sufficient}}$, ordered w.r.t.{} the slope of their vertical axis.

\begin{theorem}\label{thm:uniMonMinimum}
The uniform $2D-$monotone minimum spanning graph of a point set $P$ can be computed in $O(|P|^3)$ time. 
\end{theorem}

\begin{proof}

From Proposition~\ref{prop:events} and the previous definitions we obtain the following Proposition.

\begin{prop}\label{lem:xy-monotoneSufSlopes}
The uniform $2D-$monotone minimum spanning graph of $P$ is one of the $x'y'-$MMSG of $P$ over all Cartesian Systems $x'y'$ with $y'$ of slope in $S_{\text{sufficient}}$.
\end{prop}

We now give a $O(|P|^3)$ time rotational sweep algorithm.
The algorithm initially computes the $x_1y_1-$MMSG of $P$ and then it obtains each $x_{i+1}y_{i+1}-$MMSG of $P$ from the $x_iy_i-$MMSG of $P$.
Throughout the procedure the Cartesian System $x^{\text{opt}}y^{\text{opt}}$ in which the algorithm encountered the minimum cost solution so far is stored.
In its last step, the algorithm recomputes the $x^{\text{opt}}y^{\text{opt}}-$MMSG of $P$, which since it is equal to the rectangle of influence graph $G = (P,E)$ w.r.t.{} the Cartesian System $x^{\text{opt}}y^{\text{opt}}$ (Corollary~\ref{cor:minSpanXy}) it can be computed in $O(|P|\cdot \log |P| + |E|)$ time~\cite{OveW88}. 
The crucial proposition (which we show later) that makes the time complexity of the algorithm equal to $O(|P|^3)$ is that each transition from the $x_iy_i-$MMSG of $P$ to the $x_{i+1}y_{i+1}-$MMSG of $P$ takes $O(|P|)$ time.

For each two points $p,q$ of $P$ let \emph{$I(p,q,x_i,y_i)$} be the number of points of $P\setminus\{p,q\}$ that are included in the rectangle w.r.t.{} the Cartesian System $x_iy_i$ with opposite vertices $p$ and $q$.  
Then, $RV(q,x_i,y_i)$ can be equivalently defined using the quantities $I(p,q,x_i,y_i), p \in P\setminus\{q\}$, as follows: $p \in RV(q,x_i,y_i) \text{ if }I(p,q,x_i,y_i) = 0$.

 We store the $RV(q,x_i,y_i)$, $q \in P$, $i$ $=$ $1$,$2$, \ldots, $2l$ in the data structure $rv(q)$ which is implemented as an array of $|P|$ booleans.
We also store the $I(p,q,x_i,y_i)$, $p,q \in P$, $i$ $=$ $1$,$2$, \ldots, $2l$ in the variable $i(p,q)$.

Computing the Cartesian Systems $x_iy_i$, $i$ $=$ $1$, $2$, \ldots, $2l$ can be done in $O(|P|^2 \log |P|)$ time.
Accompanied with each Cartesian System $x_iy_i$ is the pair of points $(p_i, q_i)$ such that $\overline{p_iq_i}$ is either parallel or perpendicular to the $y_i$ axis or the $y_{i-1}$ axis.

Ichino and Sklansky~\cite{IchS85} noted that employing a range tree~\cite{BenM80,dBerCKO08} that contains the points of $P$ one can calculate i) the rectangle of influence graph of $P$, and ii) the $I(p,q,x,y), p,q, \in P$, for a Cartesian System $xy$.
Applying the previously mentioned approach, noted by Ichino and Sklansky~\cite{IchS85}, are obtained i) the rectangle of influence graph of $P$ w.r.t.{} the Cartesian System $x_1y_1$ (which by Corollary~\ref{cor:minSpanXy} equals to the $x_1y_1-$MMSG of $P$), and ii) the $I(p,q,x_1,y_1)$, $p,q, \in P$.

We now show that we can update all the $rv(p), p \in P$, such that from equal to $RV(p,x_{i-1},y_{i-1}), p \in P$, they become equal to $RV(p,x_i,y_i), p \in P$, in $O(|P|)$ total time.
For each $p \in P\setminus\{p_i,q_i\}$ the update of $rv(p)$ takes $O(1)$ time.
This is true, since only the points $p_i$ and $q_i$ have to be tested for inclusion to or removal from $rv(p)$. 
More specifically, we have to test if for one of them, say $p_i$, the rectangle with corners $p$ and $p_i$ contains (or it does not contain) $q_i$ w.r.t.{} the Cartesian System $x_iy_i$ while it did not contain (or it contained) it w.r.t.{} $x_{i-1}y_{i-1}$. 
If this is true, then the $i(p_i,p)$ changes and $p_i$ has to be tested for membership in $rv(p)$ and included to or removed from $rv(p)$.
Regarding $rv(p_i)$, the update takes $O(|P|)$ time, since for each other point $q \in P\setminus\{p_i,q_i\}$ we have to test if the rectangle with corners $q$ and $p_i$ contains (or it does not contain) $q_i$ w.r.t.{} the Cartesian System $x_iy_i$ while it did not contain it (or it contained it) w.r.t.{} the $x_{i-1}y_{i-1}$ and if so update both the $i(q,p_i)$ and the existence of $q$ in $rv(p_i)$ if necessary.
Similarly, $rv(q_i)$ can be updated in $O(|P|)$ time. 
\end{proof}

We note that the procedure of obtaining the $2D-$UMMSG can be trivially modified such that the uniform $2D$-monotone spanning graph of a point set $P$ with the least number of edges can be obtained in $O(|P|^3)$ time.
Since for an arbitrary Cartesian System $x'y'$ the $x'y'-$MMSG of $P$ is equal to the $x'y'-$monotone spanning graph of $P$ with the least number of edges (Corollary~\ref{cor:minSpanXy}), 
the only modification which is necessary is that in the transition from the Cartesian System $x_iy_i$ to the Cartesian System $x_{i+1}y_{i+1}$ we check if the  $x_{i+1}y_{i+1}-$monotone spanning graph of $P$ with the least number of edges has the least number of edges among all the produced solutions so far.

In Figure~\ref{fig:costEdgesDifference} is given a point set $P$ for which the $2D-$UMMSG of $P$ is different from the uniform $2D-$monotone spanning graph of $P$ with the least number of edges.

\begin{figure}[!htb]
\centering
\begin{tabular}{cc}
\includegraphics[scale=0.33]{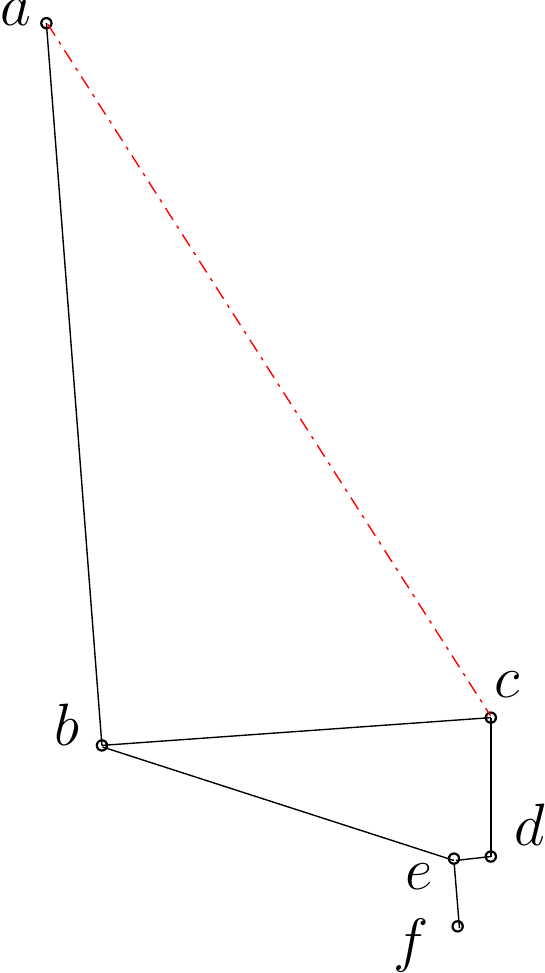} & \includegraphics[scale=0.33]{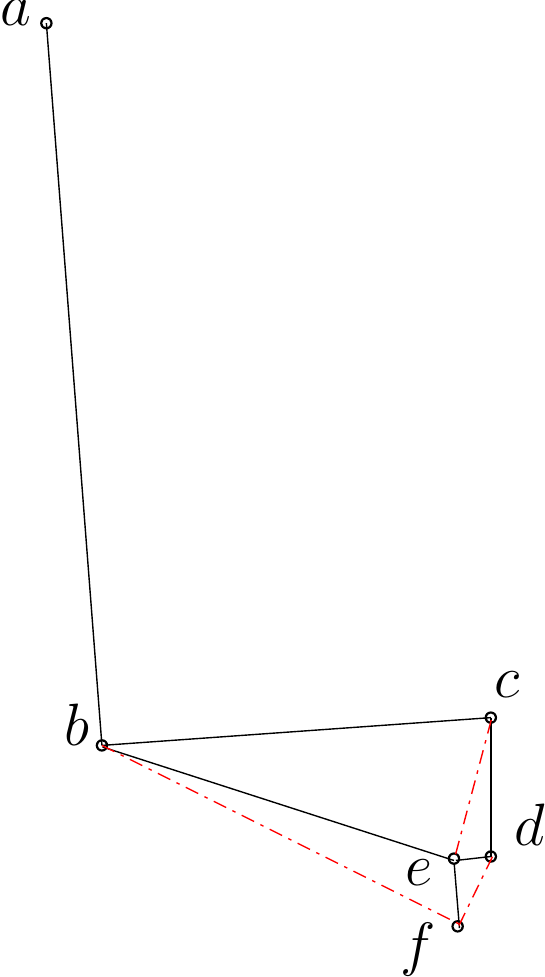} \\
(a) & (b)
\end{tabular}
\caption{The points $a$,$b$ and $c$ form a right angle. Additionally, the points $d$, $e$ and $f$ form a right angle. The slope of $\overline{de}$ is smaller than the slope of $\overline{bc}$. The uniform $2D-$monotone spanning graph with the least number of edges is obtained when the $y'$ axis becomes perpendicular to the $\overline{de}$ and is shown in (a). On the other hand the $2D-$UMMSG is obtained when the $y'$ axis becomes perpendicular to the $\overline{bc}$ and is shown in (b).}\label{fig:costEdgesDifference}
\end{figure}

In Figure~\ref{fig:exampleNonUniform} we give a point set $P$ for which the (non-uniform) $2D-$monotone spanning graph of $P$ with the least number of edges does not coincide to the (non-uniform) $2D-$monotone spanning graph of $P$ that has the minimum cost.

\begin{figure}[!htb]
\centering
\begin{tabular}{cc}
\includegraphics[scale=0.25]{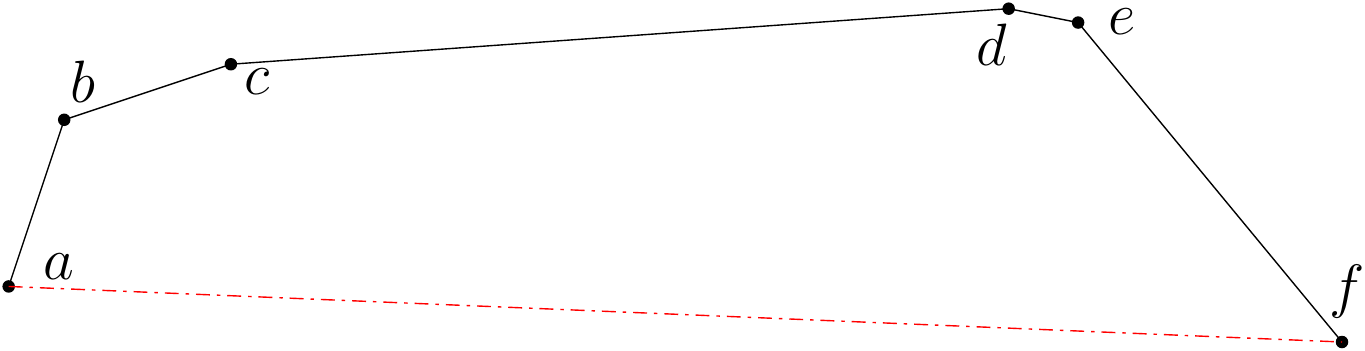} & \includegraphics[scale=0.25]{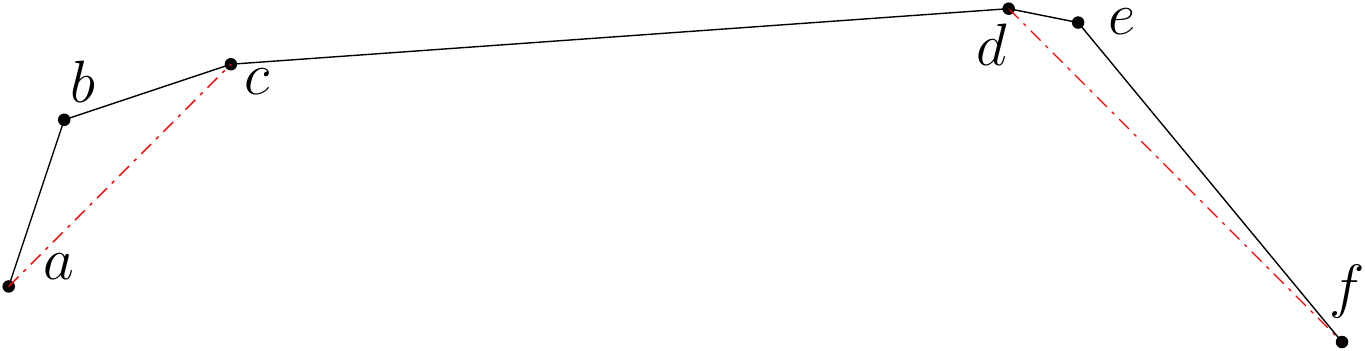} \\
(a) & (b)
\end{tabular}
\caption{The slope of $\overline{ac}$ is $\frac{\pi}{4}$ while the slope of $\overline{fd}$ is $\frac{3\pi}{4}$. In (a) is depicted the $2D-$monotone spanning graph of $P$ with the least number of edges. In (b) is illustrated the $2D-$monotone spanning graph of $P$ that has the minimum cost.}\label{fig:exampleNonUniform}
\end{figure}

Regarding recognizing uniform $2D$-monotone graphs, we note that the $O(|E|\cdot \log |P|)$ time rotational sweep algorithm given in~\cite{MasS17}, which decides if a geometric graph $G = (P,E)$ with a specified vertex $r$ as root is rooted uniform $2D-$monotone, can be easily extended into a $O(|P|^2\cdot \log|P| + |P|\cdot |E|)$ time rotational sweep algorithm that decides if $G$ is uniform $2D-$monotone.
More specifically, in order to decide if $G$ is uniform $2D-$monotone, the $|P|$ rooted geometric graphs $(p_1,G)$, $(p_2,G)$, \ldots, $(p_{|P|},G)$ where $(p_i,G)$ is the geometric graph $G$ with root $p_i$ and $\{p_1$, $p_2$, \ldots, $p_{|P|}\}$ is the vertex set of $G$, are considered.
A Cartesian System $x'y'$ is rotated counterclockwise.
From Fact~\ref{fact:recRootEvents}, it follows that one of these $|P|$ rooted geometric graphs becomes rooted $x'y'-$monotone while previously it was not, or vice versa, only when the $y'$ axis becomes (or leaves the position where it was previously) parallel or orthogonal to a line passing through two points of $P$.
Hence, $O(|P|^2)$ Cartesian Systems need to be considered, which can be computed in $O(|P|^2 \log |P|)$ time.
When the $y'$ becomes (or leaves the position that it previously was) parallel or perpendicular to a line passing through the points $p,q \in P$ then by Fact~\ref{fact:recRootEvents} the status, i.e.{} being rooted $x'y'-$monotone, of the rooted geometric graphs $(p,G)$ and $(q,G)$ may change. 
Hence, the steps of the rooted uniform 2D-monotone recognition algorithm given in~\cite{MasS17} for handling the event associated with the current Cartesian System $x'y'$ regarding the rooted geometric graphs $(p,G)$ and $(q,G)$, are applied.  
Furthermore, if $\overline{pq} \in E$ then by Fact~\ref{fact:recRootEvents} it follows that the status, i.e.{} being rooted $x'y'-$monotone, of each $(r,G), r\in P\setminus\{p,q\}$, may also change.
Hence, for each $(r,G), r\in P\setminus\{p,q\}$, the steps of the rooted uniform 2D-monotone recognition algorithm given in~\cite{MasS17} for handling the event associated with the current Cartesian System $x'y'$ are applied.
Since, the remaining steps, i.e.{} after the calculation of the sufficient axes, of the rooted uniform 2D-monotone recognition algorithm, given in~\cite{MasS17}, regarding any of these $|P|$ rooted geometric graphs take $O(|E|)$ time (Fact~\ref{fact:rU2Dm-recCost}), applying the remaining steps regarding all these $|P|$ rooted geometric graphs, takes $O(|P|\cdot |E|)$ total time.

\section{A $2-$Approximation Algorithm for the $k-$Rooted $y-$Monotone Minimum Spanning Graph Problem} \label{sec:2app-kRooted}

We now study the problem of producing the $k-$rooted $y-$monotone minimum spanning graph of a $k-$rooted point set $P$, where $1 < k < |P|$.
We assume that no two points have the same $y$ coordinate.

Let $P$ be a point set and $a,b\in \mathbb{R}$ then $P_{y > a}$ is the subset of points of $P$ whose $y$ coordinate is greater than $a$. 
Similarly are defined $P_{y \geq a}$, $P_{y < a}$ and $P_{y \leq a}$.
$P_{a < y < b}$ is the subset of points of $P$ whose $y$ coordinate is between $a$ and $b$. 
Similarly are defined $P_{a < y \leq b}$, $P_{a \leq y < b}$ and $P_{a \leq y \leq b}$.

In~\cite[Lemma 1]{MasS17} it is noted that the rooted $y-$monotone minimum spanning graph of a rooted point set $P$ with root $r$ is the union of the rooted $y-$monotone minimum spanning graphs of (i) $P_{y \leq y(r)}$ and (ii) $P_{y \geq y(r)}$.
The previous Fact is extended to the following Lemma.

\begin{lemma} \label{lem:kR_Optimal_Separation}
Let $P$ be a $k-$rooted point set, with $1 < k < |P|$, where $r_1$, $r_2$, \ldots, $r_k$ are the roots of $P$ such that $y(r_1)$ $<$ $y(r_2)$ $<$ \dots $<$ $y(r_k)$.
The $k-$rooted $y-$monotone minimum spanning graph of $P$ is the union of
\begin{enumerate}[nosep]
\item the rooted $y-$monotone minimum spanning graph of $P_{y \leq y(r_1)}$.
\item the rooted $y-$monotone minimum spanning graph of $P_{y \geq y(r_k)}$.
\item the $2-$rooted $y-$monotone minimum spanning graph of $P_{y(r_i) \leq y \leq y(r_{i+1})}$, $1 \leq i \leq k-1$.
\end{enumerate}
\end{lemma}

\begin{theorem} \label{thm:kRootedYmonMinCost}
Given a $k-$rooted point set $P$, with $1 < k < |P|$, we can obtain in $O(|P|\cdot \log^2 |P|)$ time a $k-$rooted $y-$monotone spanning graph of $P$ with cost at most twice the cost of the $k-$rooted $y-$monotone minimum spanning graph of $P$.
\end{theorem}

\begin{proof}

For a $2-$rooted point set $P$ with roots $r_1$ and $r_{2}$ that are the lowest and highest points of the point set, respectively, we prove the following Lemma.

\begin{lemma} \label{lem:2rOptConst}
Given a $2-$rooted point set $P$ with roots $r_1$ and $r_{2}$ that are the lowest and highest points of the point set, respectively, we can obtain in $O(|P|\cdot \log^2 |P|)$ time a $2-$rooted $y-$monotone spanning graph of $P$ with cost at most twice the cost of the $2-$rooted $y-$monotone minimum spanning graph of $P$.
\end{lemma}

\begin{proof}
Initially, we employ Fact~\ref{fact:ryMMSG_computation} to $P$ considering it to have only the root $r_1$ and obtain the geometric graph $G_1$.
Then, we employ Fact~\ref{fact:ryMMSG_computation} to $P$ considering it to have only the root $r_{2}$, obtaining $G_{2}$.
In the final step we return the union of $G_{1}$ and $G_{2}$. 
$G_1 \cup G_2$ is $2-$rooted $y-$monotone since $G_{1}$ ($G_{2}$) is rooted $y-$monotone with root $r_1$ (resp.{}, $r_2$).
We now show that $G_1 \cup G_2$ has cost at most twice the cost of the $2-$rooted $y-$monotone minimum spanning graph $G^{\text{opt}}$ of $P$.
Since, in $G^{\text{opt}}$ all the points $p$ are connected with $r_1$ ($r_2$) by $y-$monotone paths it follows that its cost is greater than or equal to the cost of $G_1$ (resp.{}, $G_2$).
Hence, the cost of $G_1 \cup G_2$ which is less than or equal to the sum of the costs of $G_1$ and $G_2$ is at most twice the cost of $G^{\text{opt}}$. 
\end{proof}

From Lemma~\ref{lem:kR_Optimal_Separation}, Fact~\ref{fact:ryMMSG_computation} and Lemma~\ref{lem:2rOptConst} we obtain the Theorem.
\end{proof}

A $2-$rooted planar geometric graph $G = (P,E)$ with roots $r_1$, $r_2$ s.t.{} $y(r_1) < y(p) < y(r_2), p \in P\setminus\{r_1,r_2\}$, is $2-$rooted $y-$monotone if and only if for each $p \in P\setminus\{r_1,r_2\}$ there exist $q_1, q_2 \in \text{Adj}(p)$ with $y(q_1) < y(p)< y(q_2)$~\cite{LeeP77}.
Furthermore, a rooted geometric graph $G = (P,E)$ with a (single) root $r$ that is not the highest or lowest point of $P$ is rooted $y-$monotone if and only if for each $p \in P\setminus\{r\}$ there exists $q \in \text{Adj}(p)$ such that $y(q)$ is between $y(r)$ (inclusive) and $y(p)$~\cite{MasS17}.
We extend the previous two Propositions to the following equivalent characterization of $k-$rooted $y-$monotone graphs where the latter implies an efficient recognition algorithm for $k-$rooted $y-$monotone graphs.

\begin{prop} \label{prop:k-rY-mChar}
Let $G = (P,E)$ be a $k-$rooted geometric graph, where $1 < k < |P|$, with roots $r_1$, $r_2$, \ldots, $r_k$ such that $y(r_1)$ $<$ $y(r_2)$ $<$ \dots $<$ $y(r_k)$.
$G$ is $k-$rooted $y-$monotone if and only if 
\begin{enumerate}[nosep]
\item for each $p \in P_{y < y(r_1)}$ there exists $q \in \text{Adj}(p)$ s.t.{} $y(q) \in (y(p),y(r_1)]$.
\item for each $p \in P_{y > y(r_k)}$ there exists $q \in \text{Adj}(p)$ s.t.{} $y(q) \in [y(r_k),y(p))$.
\item for each $p \in P_{y(r_i) < y < y(r_{i+1})}$ there exist $q_1, q_2 \in \text{Adj}(p)$ s.t.{} $y(q_1) \in [y(r_i),y(p))$ and $y(q_2) \in (y(p),y(r_{i+1})]$, $i$ $=$ $1$, $2$, \ldots, $k-1$.
\item there exists $q \in \text{Adj}(r_1)$ s.t.{} $y(q) \in (y(r_1), y(r_2)]$.
\item there exists $q \in \text{Adj}(r_k)$ s.t.{} $y(q) \in [y(r_{k-1}), y(r_k))$.
\item there exist $q_1, q_2 \in \text{Adj}(r_i)$ s.t.{} $y(q_1) \in [y(r_{i-1}),y(r_i))$ and $y(q_2) \in (y(r_i), y(r_{i+1})]$, $2 \leq i \leq k-1$. 
\end{enumerate}
\end{prop}

\section{Further Research Directions}

Given a point set $P$ can the $2D-$monotone spanning graph of $P$ that has the least number of edges be produced in polynomial time?

Does there exist a $t-$approximation algorithm, $t < 2$,  for the $k-$rooted $y-$monotone minimum spanning graph problem?

\noindent\textbf{Acknowledgement}: I would like to thank Professor Antonios Symvonis for his valuable contribution in developing the results presented in this paper.

\bibliographystyle{plain}
\bibliography{lib}

\end{document}